\documentclass{article}
\usepackage{amsmath,amssymb,amsthm}
\usepackage[dvipdfmx]{graphicx} % For Mac OS X
\usepackage{float}
\theoremstyle{definition}
\newtheorem{defi}{Definition}

\theoremstyle{plain}
\newtheorem{theo}{Theorem}

\newtheorem{lemm}[theo]{Lemma}

\begin{document}
\title{Independent Distributions on a Multi-Branching AND-OR Tree of Height 2
}
\author{Mika Shigemizu${}^{1}$, Toshio Suzuki${}^{2}$\thanks{Corresponding author. This  work  was  partially  supported  by  Japan  Society  for  the Promotion of Science (JSPS) KAKENHI (C) 16K05255.} and Koki Usami${}^{3}$
\\
Department of Mathematical Sciences, %\\
Tokyo Metropolitan University, \\ 
Minami-Ohsawa, Hachioji, Tokyo 192-0397, Japan\\
1: mksgmn$\_$dr5@yahoo.co.jp
\quad
2: toshio-suzuki@tmu.ac.jp
\\
3: cocokidream@yahoo.co.jp
} 

\date{\today}

\maketitle              % typeset the title of the contribution

\begin{abstract}
We investigate an AND-OR tree $T$ and a probability distribution $d$ 
on the truth assignments to the leaves. 
Tarsi (1983) showed that if $d$ is an independent and identical distribution (IID) 
such that probability of a leaf having value 0 is neither 0 nor 1 
then, under a certain assumptions, there exists an optimal 
algorithm that is depth-first. 
We investigate the case where $d$ is an independent distribution (ID) 
and probability depends on each leaf. 
It is known that in this general case, if height is greater than or equal to 3, 
Tarsi-type result does not hold. 
It is also known that for a complete binary tree of height 2, Tarsi-type result 
certainly holds. 
In this paper, we ask whether Tarsi-type result holds for an AND-OR tree of 
height 2. Here, a child node of the root is either an OR-gate or a leaf: 
The number of child nodes of an internal node is arbitrary, 
and depends on an internal node. 
We give an affirmative answer. 
Our strategy of the proof is to reduce the problem to the case of 
directional algorithms. We perform induction on the number of leaves, 
and modify Tarsi's method to suite height 2 trees. 
We discuss why our proof does not apply to height 3 trees.

\vspace{\baselineskip}


MSC[2010] 68T20; 68W40
\end{abstract}

%%%%%%%%%%%%%%%%%%%%%%%%%%%%%%%%%%%%%%%%%%%%%%%%%%%%%%%%%%%%%%%%%%%%%
\section{Introduction}
%%%%%%%%%%%%%%%%%%%%%%%%%%%%%%%%%%%%%%%%%%%%%%%%%%%%%%%%%%%%%%%%%%%%%

Depth-first algorithm is a well-known type of tree search algorithm. 
Algorithm $A$ on a tree $T$ is \emph{depth-first} if the following holds for each internal node $x$ of $T$: 
Once $A$ probes a leaf that is a descendant of $x$, $A$ does not probe leaves 
that are not descendant of $x$ until $A$ finds value of $x$. 

With respect to analysis of algorithm, the concept of depth-first algorithm has an advantage that it is well-suited for induction on subtrees. An example of such type of induction may be found in our former paper \cite{SN15}. 

Thus, given a problem on a tree, it is theoretically interesting question 
to ask whether there exists an optimal algorithm that is depth-first. 
Here, cost denotes (expected value of) the number of leaves probed during 
computation, and an algorithm is optimal if it achieves the minimum cost 
among algorithms considered in the question. 

If an associated evaluation function of a mini-max tree is bi-valued 
and the label of the root is MIN (MAX, respectively), the tree is equivalent to 
an AND-OR tree (an OR-AND tree). 
In other words, the root is labeled by AND (OR), and 
AND layers and OR layers alternate. 
Each leaf has a truth value 0 or 1, where we identify 0 with false, 
and 1 with true. 

A fundamental result on optimal algorithms on 
an AND-OR trees is given by Tarsi. 
A tree is \emph{balanced} (in the sense of Tarsi) if 
(1) any two internal nodes of the same depth (distance from the root) 
have the same number of child nodes, and 
(2) all of the leaves have the same depth.
A probability distribution $d$ on the truth assignments to the leaves 
is an \emph{independent distribution} (ID) 
if the probability of each leaf having value 0 depends on 
the leaf, and values of leaves are determined independently. 
If, in addition, all the probabilities of the leaves are the same 
(say, $p$), 
$d$ is an \emph{independent and identical distribution} (IID). 
Algorithm $A$ is \emph{directional} \cite{Pe80} if there is a fixed linear order 
of the leaves and for any truth assignment, the order of probing by $A$ 
is consistent with the order. 

The result of Tarsi is as follows. 
Suppose that AND-OR tree $T$ is balanced and $d$ is an IID such that 
$p\ne 0, 1$. Then there exists an optimal algorithm that is depth-first 
and directional \cite{Ta83}. 

This was shown by an elegant induction. 
For an integer $n$ such that $0 \leq n \leq h = $ (the height of the tree), 
an algorithm $A$ is called \emph{$n$-straight} if for each node $x$ 
whose distance from the leaf is at most $n$, 
once $A$ probes a leaf that is a descendant of $x$, 
 $A$ does not probe leaves 
that are not descendant of $x$ until $A$ finds value of $x$. 
The proof by Tarsi is induction on $n$. 
Under an induction hypothesis, we take an optimal algorithm $A$ that is 
$(n-1)$-straight but not $n$-straight. 
Then we modify $A$ and get two new algorithms. 
By the assumption, 
the expected cost by $A$ is not greater than the costs by the modified 
algorithms, thus we get inequalities. By means of the inequalities, 
we can eliminate a non-$n$-straight move from $A$ 
without increasing cost. 

In the above mentioned proof by Tarsi, derivation of the inequalities 
heavily depends on the hypothesis that the distribution is an IID. 
In the same paper, Tarsi gave an example of an ID on an AND-OR 
tree of the following properties: 
The tree is of height 4, not balanced, and no optimal algorithm is 
depth-first. Later, we gave another example of such an ID where a tree is 
balanced, binary and height 3 \cite{Su17b}. 

On the other hand, as is observed in \cite{Su17b}, 
in the case of a balanced binary AND-OR tree of height 2, 
Tarsi-type result holds for IDs in place of IIDs. 

\begin{table}[ht]
\begin{center}
\begin{tabular}{|p{.25\textwidth}||p{.28\textwidth}|p{.28\textwidth}|}
\hline
 & ID & IID \\
\hline
\hline
height 2, binary & Yes. S. \cite{Su17b} & \\ \cline{1-2}
height 2, general &  & \\ \cline{1-2}
height $\geq$ 3 & No. S. \cite{Su17b} & Yes. Tarsi \cite{Ta83} \\
 & (see also Tarsi \cite{Ta83}) & \\ 
\hline
\end{tabular}
\end{center}
\caption{Existence of an optimal algorithm that is depth-first}
\label{table:1}
\end{table}%

Table~\ref{table:1} summarizes whether Tarsi-type result holds or not. 
In the table, we assume that an AND-OR tree is balanced, and that 
the probability of each leaf having value 0 is neither 0 nor 1. 

In this paper, we ask whether Tarsi-type result holds for 
the case where a tree is height 2 and the number of child nodes is 
arbitrary. We give an affirmative answer. 

We show a slightly stronger result. 
We are going to investigate a tree of the following properties. 
The root is an AND-gate, 
and a child node of the root is either an OR-gate or a leaf. 
The number of child nodes of an internal node is arbitrary, 
and depends on an internal node. 
Figure~\ref{fig:orandtreeh2n1kai} is an example of such a tree. 
Now, suppose that an ID on the tree is given and that at each leaf, 
the probability of having value 0 is neither 0 nor 1. 
Under these mild assumptions, we show that there exists an optimal 
algorithm that is depth-first and directional. 

%%%%%
\begin{figure}[hb] 
\centering
\includegraphics[width=.75\textwidth]{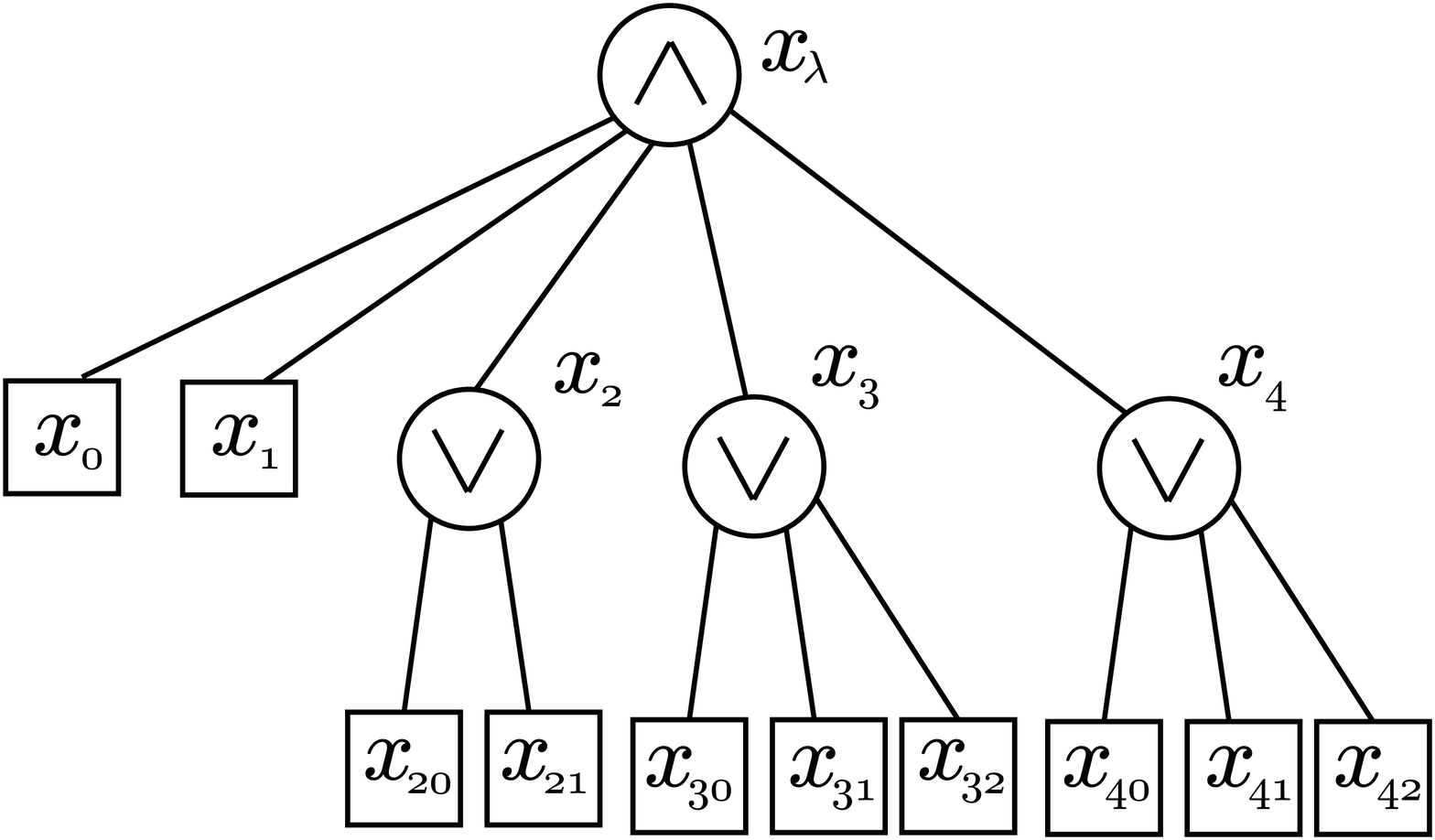}
\caption{Example of a height 2 AND-OR tree that is not balanced}
\label{fig:orandtreeh2n1kai}
\end{figure}
%%%%%

Our strategy of the proof is to reduce the problem to the case of directional algorithms. 
We perform induction on the number of leaves, and modify Tarsi's method to go along with properties particular to hight 2 trees. 
The first author and the third author showed, in their talk \cite{SU18}, 
a restricted version of the present result in which only directional (possibly non-depth-first) algorithms are taken into consideration. 
In the talk, the core of the strategy is suggested by the first author. 
The second author reduced the general case, in which  non-directional algorithms are taken into consideration, to the case of directional algorithms. 

The paper \cite{Su17a} gives an exposition of the background. 
It is a short survey on the work of Liu and Tanaka 
\cite{LT07} and its subsequent developments \cite{SN12,SN15,POLT16}. 
The paper \cite{PPNTY17} is also in this line. 
Some classical important results by 1980's may be 
found in the papers \cite{KM75,Pe80,Pe82,Ta83} and \cite{SW86}. 

We introduce notation in section~\ref{section:preliminaries}. 
We show our result in section~\ref{section:results}. 
In section~\ref{section:summary}, we discuss why our proof 
does not apply to the case of height 3, and discuss future directions.  

%%%%%%%%%%%%%%%%%%%%%%%%%%%%%%%%%%%%%%%%%%%%%%%%%%%%%%%%%%%%%%%%%%%%%
\section{Preliminaries} \label{section:preliminaries}
%%%%%%%%%%%%%%%%%%%%%%%%%%%%%%%%%%%%%%%%%%%%%%%%%%%%%%%%%%%%%%%%%%%%%

If $T$ is a balanced tree (in the sense of Tarsi, see Introduction) 
and there exists a positive integer $n$ such that all of the internal nodes 
have exactly $n$ child nodes, we say $T$ is a \emph{complete $n$-ary tree}. 

For algorithm $A$ and distribution $d$, we denote (expected value of) 
the cost by $\mathrm{cost} (A,d)$. 

We are interested in a multi-branching AND-OR tree of height 2, 
where a child node of the root is either a leaf or an OR-gate, 
and the number of leaves depend on each OR-gate. 
For simplicity of notation, we investigate a slightly larger class 
of trees. 

Hereafter, by ``a multi-branching AND-OR tree of height at most 2'', 
we denote an AND-OR tree $T$ of the following properties. 

\begin{itemize}
\item The root is an AND-gate. 
\item We allow an internal node to have only one child, 
provided that the tree has at least two leaves. 
\item All of the child nodes of the root are OR-gates. 
\end{itemize}

The concept of ``multi-branching AND-OR tree of height at most 2'' include 
the multi-branching AND-OR trees of height 2 in the original sense 
(because an OR-gate of one leaf is equivalent to a leaf), 
the AND-trees of height 1 
(this case is achieved when all of the OR-gates have one leaf) 
and the OR-trees of height 1 
(this case is achieved when the root has one child). 
The simplest case is a tree of just two leaves, 
and this case is achieved exactly in either of the following two. 
(1) The tree is equivalent to a binary AND-tree of height 1: 
(2) The tree is equivalent to a binary OR-tree of height 1. 

Suppose that $T$ is a multi-branching AND-OR tree of height at most 2. 

\begin{itemize}
\item We let $x_{\lambda}$ denote the root. 
\item By $r$ we denote the number of child nodes of the root. 
$x_{0}, \dots, x_{r-1}$ are the child nodes of the root. 
\item For each $i$ ($0 \leq i < r$), we let $a(i)$ denote 
the number of child leaves of $x_{i}$. 
$x_{i,0}, \dots, x_{i,a(i)-1}$ are the child leaves of $x_{i}$.
\end{itemize}

Figure~\ref{fig:orandtreeh2n1} is an example of such a tree, 
where $r=5$, $a(0)=1$ and $a(4)=3$. 
The tree is equivalent to the tree in Figure~\ref{fig:orandtreeh2n1kai}. 

%%%%%
\begin{figure}[H] 
\centering
\includegraphics[width=.75\textwidth]{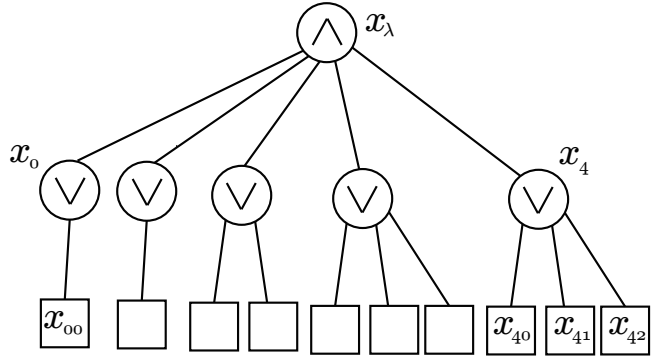}
\caption{Example of an AND-OR tree of height at most 2}
\label{fig:orandtreeh2n1}
\end{figure}
%%%%%

Suppose that $d$ is an ID on $T$. 
For each $i$ ($0 \leq i < r$) and $j$ ($0 \leq j < a(i)$), 
we use the following symbols. 

\begin{itemize}
\item $p(i,j)$ is the probability of $x_{i,j}$ having value 0. 
\item $p(i)$ is the probability of $x_{i}$ having value 0.
\item Since $d$ is an ID, its restriction to the child nodes of 
$x_{i}$ is an ID on the subtree whose root is $x_{i}$. 
Here, we denote it by the same symbol $d$. 
Then we define $c(i)$ as follows. 
\[
c(i) = \min_{A} \mathrm{cost}(A,d),
\] 
\noindent
where $A$ runs over algorithms finding value of $x_{i}$, 
and $\mathrm{cost}(A,d)$ is expected cost. 
\end{itemize}

Thus, $p(i)$ is the product $p(i,0)\cdots p(i,a(i)-1)$. 
If $a(i)=1$ then $c(i)=1$. 
If $a(i)\geq 2$ and we have $p(i,0) \leq \dots \leq p(i,a(i)-1)$, then 
$c(i) =1+p(i,0)+p(i,0)p(i,1)+\dots +p(i,0)\cdots p(i,a(i)-1)$. 

Tarsi \cite{Ta83} investigated a depth-first algorithm that probes the leaves 
from left to right, skipping a leaf whenever there is sufficient information 
to evaluate one of its ancestors, and he called it $\mathrm{SOLVE}$. 
We investigate a similar algorithm depending on a given independent 
distribution. 

%%%%%
\begin{defi} 
Suppose that $T$ is a multi-branching AND-OR tree of height at most 2. 

\begin{enumerate}
\item 
Suppose that $d$ is an ID on $T$ and for each $i$ ($0 \leq i < r$) 
and $j$ ($0 \leq j < a(i)$), we have $p(i,j) \neq 0,1$. 

By $\mathrm{SOLVE}_{d}$, we denote the unique depth-first directional 
algorithm such that the following hold for all $i,s, j, k$ 
($0 \leq i < r$, $0 \leq s < r$, $0 \leq j < a(i)$, $0 \leq k < a(i)$). 
 \begin{enumerate}
 \item If $c(i)/p(i) < c(s)/p(s)$ then priority of (probing the leaves under) $x_{i}$ 
is higher than that of $x_{s}$. 
 \item If $c(i)/p(i) = c(s)/p(s)$ and $i<s$ then 
priority of $x_{i}$ is higher than that of $x_{s}$. 
 \item If $p(i,j) < p(i,k)$ then priority of (probing) $x_{i,j}$ is higher than $x_{i,k}$. 
 \item If $p(i,j) = p(i,k)$ and $j<k$ then priority of $x_{i,j}$ is higher than $x_{i,k}$. 
 \end{enumerate}
\item 
Suppose that we remove some nodes (except for the root of $T$) from $T$, and if a removed node has descendants, we remove them too. Let $T^{\prime}$ be the resulting tree. 
Suppose that $\delta$ is an ID on $T^{\prime}$ and for each $i$ ($0 \leq i < r$) 
and $j$ ($0 \leq j < a(i)$) such that $x(i,j)$ is a leaf of $T^{\prime}$, 
we have $p(i,j) \neq 0,1$. 

By $\mathrm{SOLVE}^{T}_{\delta}$, we denote the unique depth-first directional 
algorithm of the following properties. 
For all $i,s, j, k$ 
($0 \leq i < r$, $0 \leq s < r$, $0 \leq j < a(i)$, $0 \leq k < a(i)$), 
if $x_{i}$ and $x_{s}$ are nodes of $T^{\prime}$, the above-mentioned assertions (a) and (b) hold; 
and if $x_{i,j}$ and $x_{i,k}$ are leaves of $T^{\prime}$, the above-mentioned assertions (c) and (d) hold. 
\end{enumerate}
\end{defi}
%%%%%

%%%%%
\begin{lemm} \label{lemm:1}
Suppose that $T$ is a multi-branching AND-OR tree of height at most 2. 
Suppose that $d$ is an ID on $T$ and for each $i$ $(0 \leq i < r)$ 
and $j$ $(0 \leq j < a(i) )$, we have $p(i,j) \neq 0,1$. 
Then $\mathrm{SOLVE}_{d}$ achieves the minimum cost among 
depth-first directional algorithms. 
To be more precise, if $A$ is a depth-first directional algorithm 
then $\mathrm{cost} (\mathrm{SOLVE}_{d},d) \leq \mathrm{cost}(A,d)$. 
\end{lemm}

\begin{proof}
It is straightforward.
\end{proof}
%%%%%

%%%%%%%%%%%%%%%%%%%%%%%%%%%%%%%%%%%%%%%%%%%%%%%%%%%%%%%%%%%%%%%%%%%%%
\section{Result} \label{section:results}
%%%%%%%%%%%%%%%%%%%%%%%%%%%%%%%%%%%%%%%%%%%%%%%%%%%%%%%%%%%%%%%%%%%%%

%%%%%
\begin{theo} \label{theo:main}
Suppose that $T$ is a multi-branching AND-OR tree of height at most 2. 
Suppose that $d$ is an ID on $T$ and for each $i$ $(0 \leq i < r)$ 
and $j$ $(0 \leq j < a(i) )$, we have $p(i,j) \neq 0,1$. 
Then $\mathrm{SOLVE}_{d}$ achieves the minimum cost among all of the algorithms 
(depth-first or non-depth-first, directional or non-directional). 
Therefore, there exists a depth-first directional algorithm that is 
optimal among all of the algorithms.
\end{theo}

\begin{proof}
We perform induction on the number of leaves. The base cases are 
the binary AND-trees of height 1 and the binary OR-trees of height 1. 
In general, if $T$ is equivalent to a tree of height 1, 
the assertion of the theorem clearly holds. 

To investigate induction step, we assume that $T$ has at least 
three leaves. Our induction hypothesis is 
that for any multi-branching AND-OR tree $T^{\prime}$ of height at most 2, 
if the number of leaves of $T^{\prime}$ is less than that of $T$ 
then the assertion of theorem holds for $T^{\prime}$. 

We fix an algorithm $A$ that minimizes $\mathrm{cost}(A,d)$ 
among all of the algorithms 
(depth-first or non-depth-first, directional or non-directional). 

Case 1: At the first move of $A$, $A$ makes a query to a leaf $x_{i,0}$ 
such that $a(i)=1$. 
In this case, if $A$ finds that $x_{i,0}$ has value 0 then $A$ returns 0 and finish. 
Otherwise, $A$ calls a certain algorithm (say, $A^{\prime}$) 
on $T - x_{i}$, that is, the tree given by removing 
$x_{i}$ and $x_{i,0}$ from $T$. 
The probability distribution given by restricting $d$ 
to $T - x_{i}$ is an ID, and a probability of any leaf is neither 0 nor 1.

Therefore, by induction hypothesis, without loss of generality, 
we may assume that $A^{\prime}$ is a depth-first directional algorithm 
on $T - x_{i}$. Therefore, $A$ is a depth-first directional algorithm. 
Hence, by Lemma~\ref{lemm:1}, 
the same cost as $A$ is achieved by $\mathrm{SOLVE}_{d}$. 

Case 2: Otherwise. At the first move of $A$, 
$A$ makes a query to a leaf $x_{i,j}$ such that $a(i)\geq 2$. 

Let $T_{0}:=T - x_{i,j}$, the tree given by removing $x_{i,j}$ from $T$. 
In addition, let $T_{1}:=T - x_{i}$, the tree given by removing $x_{i}$ and all of the leaves under $x_{i}$ from $T$. 
Here, $T_{0}$ and $T_{1}$ inherit all of the indices (for example, ``$3,1$'' of $x_{3,1}$) from $T$. 

If $T_{1}$ is empty then $T$ is equivalent to a tree of height 1, and this case 
reduces to our observation in the base case. Thus, throughout rest of Case 2, we assume that $T_{1}$ is non-empty. 

If $A$ finds that $x_{i,j}$ has value 0 then 
$A$ calls a certain algorithm (say, $A_{0}$) on $T_{0}$. 

If $A$ finds that $x_{i,j}$ has value 1 then 
$A$ calls a certain algorithm (say, $A_{1}$) 
on $T_{1}$. 

For each $s=0,1$, let $d[s]$ be the restriction of $d$ to $T_{s}$. 
In the same way as Case 1, without loss of generality, 
we may assume that $A_{s}$ is $\mathrm{SOLVE}^{T}_{d[s]}$ on $T_{s}$. 

Hence, there is a permutation 
$X=\langle x_{i,s(0)}, \dots, x_{i,s(a(i)-2)}\rangle$ 
of the leaves under $x_{i}$ except $x_{i,j}$, 
and there are possibly empty sequences of leaves, 
$Y=\langle y_{0}, \dots, y_{k-1}\rangle$ and 
$Z=\langle z_{0}, \dots, z_{m-1}\rangle$, with the following properties. 

\begin{itemize}
\item The three sets 
$\{ x_{i} \}$, 
$Y^{\ast} = \{ y : y$ is a parent of a leaf in $Y \}$, and 
$Z^{\ast} = \{ z : z$ is a parent of a leaf in $Z \}$ 
are mutually disjoint, 
and their union equals $\{ x_{0}, \dots, x_{r-1} \}$, 
the set of all child nodes of $x_{\lambda}$. 

\item 
The search priority of $A_{0}$ is in accordance with $YXZ$ (thus, $y_{0}$ is the first). 

\item 
The search priority of $A_{1}$ is in accordance with $YZ$. 
\end{itemize}

Case 2.1: $Z$ is non-empty. 

We are going to show that $Y$ is empty. Assume not. Let $B$ be 
the depth-first directional algorithm on $T$ 
whose search priority is $Y~x_{i,j}~XZ$. 

Let $A_{Y}$ ($A_{X}, A_{Z}$, respectively) denote 
the depth-first directional algorithm on the subtree above by $Y$ ($X,Z$), 
where search priority is in accordance with $Y$ ($X,Z$). 
Thus, we may write $A_{0}$ as ``$A_{Y}$; $A_{X}$; $A_{Z}$'', 
$A_{1}$ as ``$A_{Y}$; $A_{Z}$'', and 
$B$ as ``$A_{Y}$; Probe $x_{i,j}$; $A_{X}$; $A_{Z}$''.

We look at the following events. 
Recall that $Y^{\ast}$ is the set of all parent nodes of leaves in $Y$. 
 
$E_{Y}$: ``At least one element of $Y^{\ast}$ has value 0.''

$E_{X}$: ``All of the elements of $X$ have value 0.''

Since the tree is height 2 and $Z$ is non-empty, in each of $A_{0}$, $A_{1}$ and $B$, the following holds: 
``$A_{Y}$ finds value of $x_{\lambda}$ if and only if $E_{Y}$ happens.'' 

In the same way, in $A_{0}$ and in $B$, under assumption that $A_{X}$ is called, 
$A_{X}$ finds value of the root if and only if $E_{X}$ happens. 

Thus, flowcharts of $A$ and $B$ as Boolean decision trees are as described in Figure~\ref{fig:flowchart02} and Figure~\ref{fig:flowchartb1}, respectively. 

%%%%%
\begin{figure}[H]
\centering
\includegraphics[width=.75\textwidth]{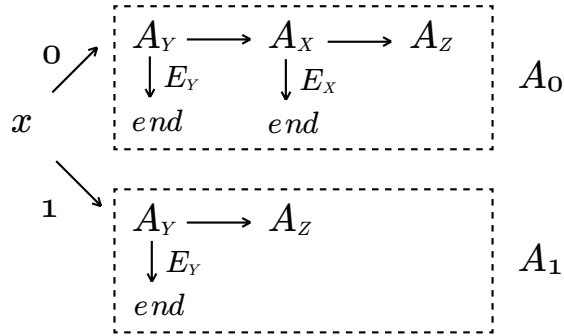}
\caption{Flowchart of $A$ (Case 2.1, in the presence of $Y$)}
\label{fig:flowchart02}
\end{figure}
\begin{figure}[H]
\centering
\includegraphics[width=.54\textwidth]{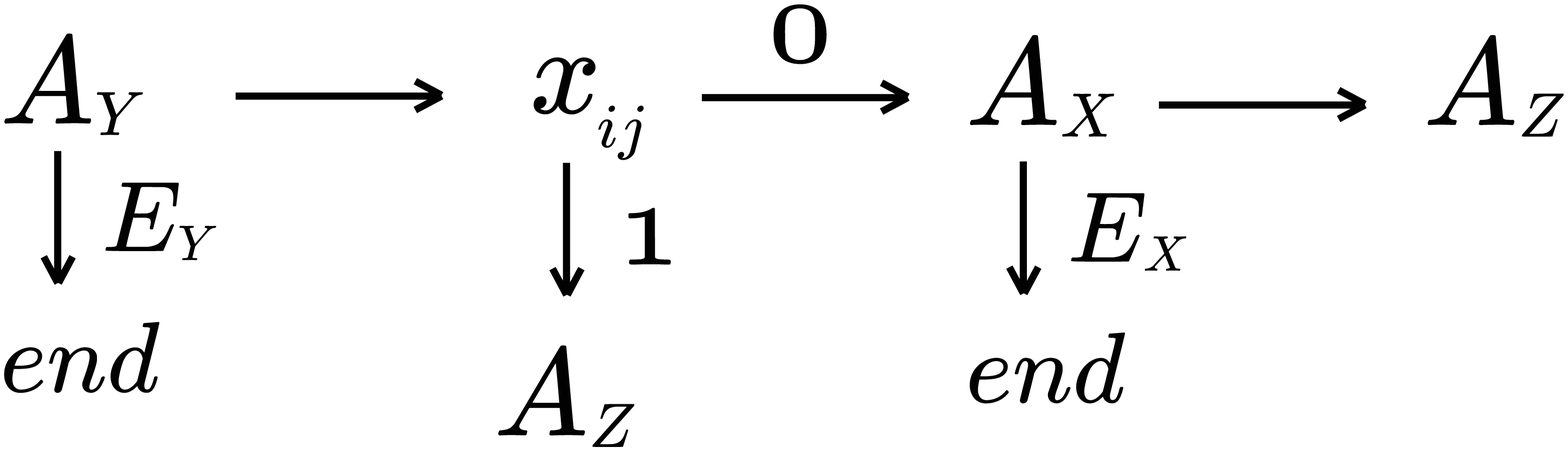}
\caption{Flowchart of $B$ (in the presence of $Y$)}
\label{fig:flowchartb1}
\end{figure}
%%%%%

Therefore, letting $p_{Y} = \mathrm{prob} [\neg E_{Y}] $ (that is, probability of the negation) and 
$p_{X} = \mathrm{prob} [\neg E_{X}]$, 
the cost of $A$ and $B$ are as follows. 
In the following formulas, $C_Y$ denotes 
$\mathrm{cost}(A_{Y}, d_{Y})$, 
where $d_{Y}$ denotes the probability distribution 
given by restricting $d$ to $Y$. 
$C_X$ and $C_Z$ are similarly defined. 

\begin{eqnarray}
& & \mathrm{cost}(A,d) \notag \\
&=& 1 + p(i,j) \{ C_Y + p_{Y} (C_X + p_{X} C_Z) \} + (1-p(i,j)) (C_Y+p_{Y}C_Z) 
\notag \\
&=& 1 + C_Y + p(i,j)p_{Y}(C_X+p_{X}C_Z) + (1-p(i,j))p_{Y}C_Z 
\end{eqnarray}

\begin{eqnarray}
& & \mathrm{cost}(B,d) \notag \\
&=& C_Y + p_{Y} [ 1 + \{ p(i,j)(C_X+p_{X}C_Z) + (1 - p(i,j))C_Z \} ]
\notag \\
&=& C_Y + p_{Y} + p_{Y}p(i,j)(C_X+p_{X}C_Z) + p_{Y}(1-p(i,j))C_Z 
\end{eqnarray}

Therefore, $\mathrm{cost}(A,d) - \mathrm{cost}(B,d) = 1 - p_{Y}$. 
However, by our assumption on $d$ that probability of each leaf (having value 0) 
is neither 0 nor 1, $E_{Y}$ has positive probability, thus $p_{Y} < 1$. 
Thus $\mathrm{cost}(A,d) - \mathrm{cost}(B,d)$ is positive, 
and this contradicts to the assumption that $A$ achieves the minimum cost. 

Hence, we have shown that $Y$ is empty. Therefore, $A$ is the following algorithm 
(Figure~\ref{fig:flowchart12}): 
``Probe $x_{i,j}$. If $x_{i,j}=0$ then perform depth-first directional search 
on $T_{0}=T - x_{i,j}$, where search priority is in accordance with $XZ$. 
Otherwise (that is, $x_{i,j}=1$), perform depth-first directional search 
on $T_{1}=T - x_{i}$, where search priority is in accordance with $Z$.''

%%%%%
\begin{figure}[H]
\centering
\includegraphics[width=.6\textwidth]{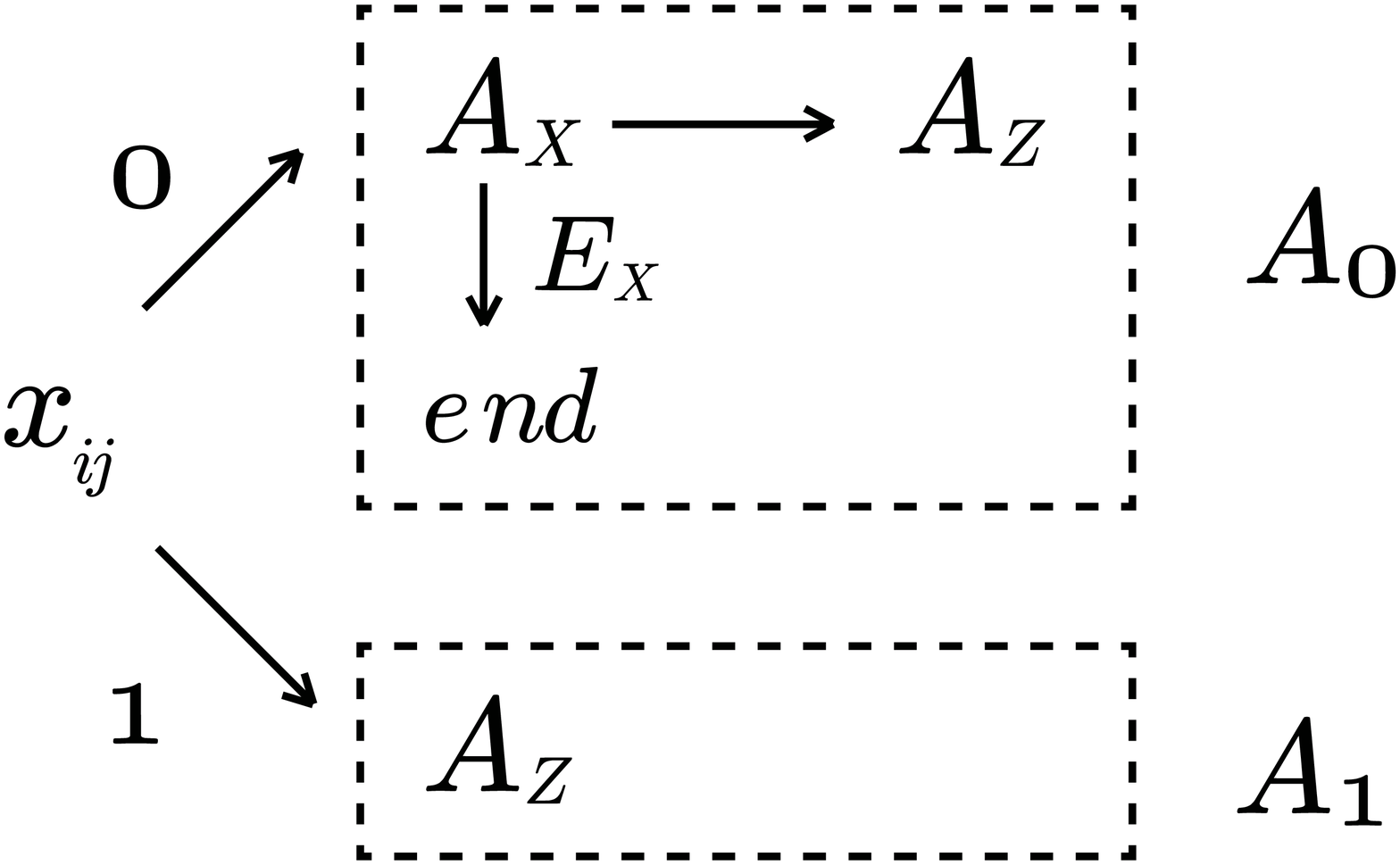}
\caption{Flowchart of $A$ (in the absence of $Y$)}
\label{fig:flowchart12}
\end{figure}
%%%%%

Thus, $A$ is a depth-first directional algorithm. 
Hence, by Lemma~\ref{lemm:1}, 
the same cost as $A$ is achieved by $\mathrm{SOLVE}_{d}$. 

Case 2.2: Otherwise. In this case, $Z$ is empty. The proof is similar to 
Case 2.1.
\end{proof}
%%%%%

%%%%%%%%%%%%%%%%%%%%%%%%%%%%%%%%%%%%%%%%%%%%%%%%%%%%%%%%%%%%%%%%%%%%%
\section{Concluding remarks} \label{section:summary}
%%%%%%%%%%%%%%%%%%%%%%%%%%%%%%%%%%%%%%%%%%%%%%%%%%%%%%%%%%%%%%%%%%%%%

%%%%%%%%%%%%%%%%%%%%%%%%%%%%%%%%%%%%%%%%%%%%%%%%%%%%%%%%%%%%%%%%%%%%%
\subsection{Difference between height 2 case and heihgt 3 case}
%%%%%%%%%%%%%%%%%%%%%%%%%%%%%%%%%%%%%%%%%%%%%%%%%%%%%%%%%%%%%%%%%%%%%

As is mentioned in Introduction, the counterpart to Theorem~\ref{theo:main} 
does not hold for the case of height 3. 
We are going to discuss why the proof of Theorem~\ref{theo:main} does not 
work for the case of height 3. 

Figure~\ref{fig:orandtreeh3n1} is a complete binary OR-AND tree of height 3. 

%%%%%
\begin{figure}[H]
\centering
\includegraphics[width=.7\textwidth]{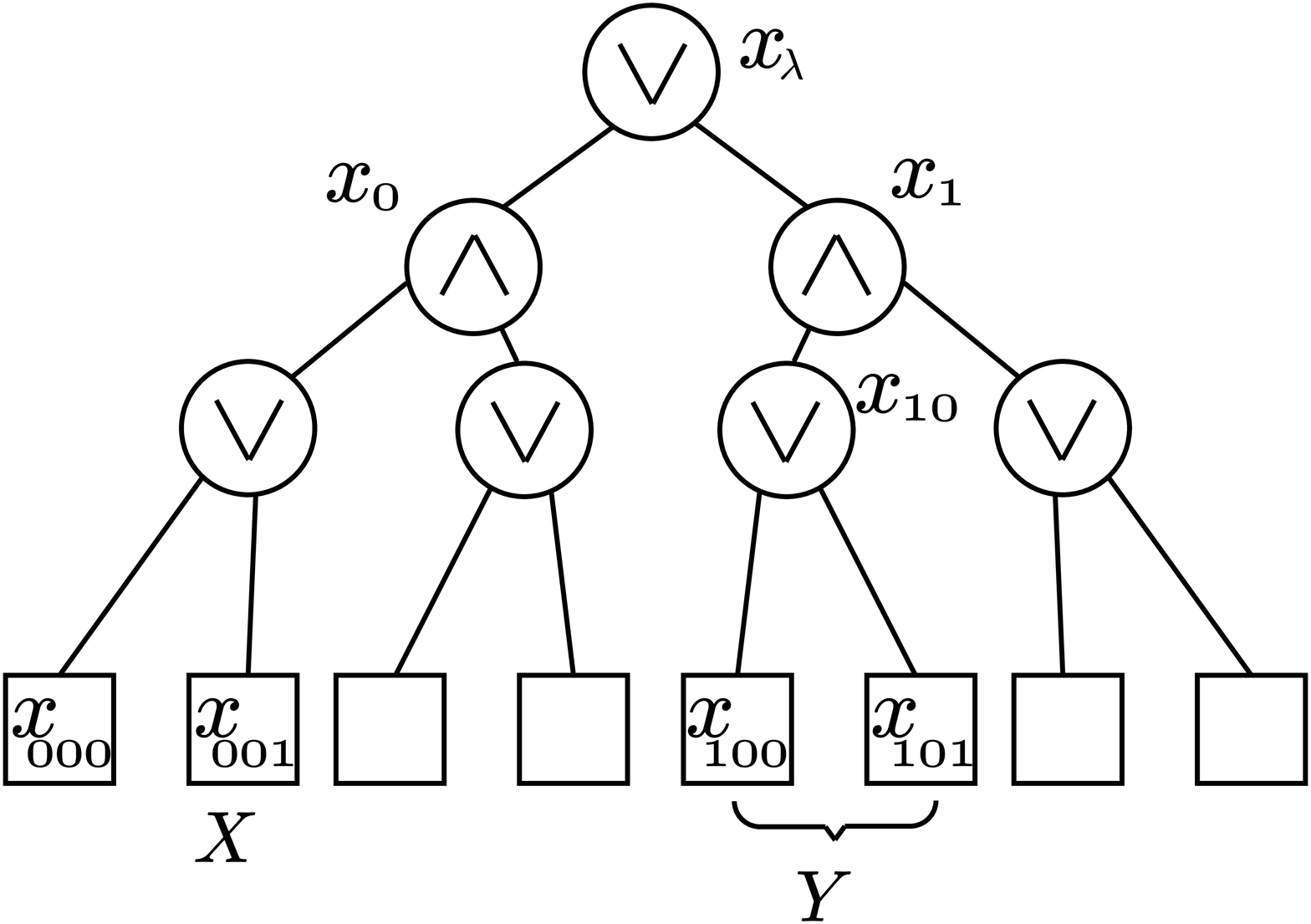}
\caption{A binary OR-AND tree of height 3}
\label{fig:orandtreeh3n1}
\end{figure}
%%%%%

Suppose that algorithm $A$ is as follows. 
Let $Y=\langle x_{100}, x_{101} \rangle$, $X=\langle x_{001} \rangle$, and 
$Z=\langle x_{010}, x_{011}, x_{110}, x_{111}\rangle$. 
At the first move, $A$ probes $x_{000}$. 
If $x_{000} = 0$ then $A$ probes in accordance with order $YXZ$. 
If $x_{000} = 1$ then $A$ probes in accordance with order $YZ$. 

Let $A^{\prime}_{Y}$ be the algorithm on $Y$ such that $x_{100}$ has higher priority 
of probing than $x_{101}$. 

Suppose that an ID $d$ on the tree is given, and that, at each leaf, 
probability of having value 0 is neither 0 nor 1. 
We investigate the following event. 

$E^{\prime}_{Y}$: ``$x_{10}$ has value 0.''

On the one hand, by our assumption on ID $d$, $E^{\prime}_{Y}$ has positive probability. 
On the other hand, whether $E^{\prime}_{Y}$ happens or not, $A_{Y}$ does not find value of $x_{\lambda}$. In other words, probability of ``$A^{\prime}_{Y}$ finds value of $x_{\lambda}$'' is 0. 
Therefore, $E^{\prime}_{Y}$ is not equivalent to the assertion 
``$A^{\prime}_{Y}$ finds value of $x_{\lambda}$''. 
Hence, counterpart to our observation in Case 2.1 of Theorem~\ref{theo:main} 
does not work for the present setting.

%%%%%%%%%%%%%%%%%%%%%%%%%%%%%%%%%%%%%%%%%%%%%%%%%%%%%%%%%%%%%%%%%%%%%
\subsection{Summary and future directions}
%%%%%%%%%%%%%%%%%%%%%%%%%%%%%%%%%%%%%%%%%%%%%%%%%%%%%%%%%%%%%%%%%%%%%

Given a tree $T$, let $\mathrm{IID}_{T}^{+}$ ($\mathrm{ID}_{T}^{+}$, respectively) 
denote the set of all IIDs on $T$ (IDs on $T$) such that, at each leaf, 
probability having value 0 is neither 0 nor 1.
Now we know the following. 

\begin{enumerate}
\item (Tarsi \cite{Ta83}) Suppose that $T$ is a balanced AND-OR tree of any height, 
and that $d \in \mathrm{IID}_{T}^{+}$. 
Then there exists an optimal algorithm that is depth-first and directional.  

\item (S. \cite{Su17b}) Suppose that $T$ is a complete binary OR-AND tree of height 3. 
Then there exists $d \in \mathrm{ID}_{T}^{+}$ such that 
no optimal algorithm is depth-first. 
 
\item (Theorem 2) Suppose that $T$ is an AND-OR tree of height 2, 
and that $d \in \mathrm{ID}_{T}^{+}$. 
Then there exists an optimal algorithm that is depth-first and directional. 
\end{enumerate}

Suppose that $T$ is a complete binary AND-OR tree of height $h \geq 3$. 
There is yet some hope to find a subset $\mathcal{D}$ of $\mathrm{ID}_{T}^{+}$ 
of the following properties. 

\begin{itemize}
\item $\mathrm{IID}_{T}^{+} \subsetneq \mathcal{D} \subsetneq \mathrm{ID}_{T}^{+}$
\item For each $d \in \mathcal{D}$, there exists an optimal algorithm that is depth-first and directional. 
\end{itemize}

%%%%%%%%%%%%%%%%%%%%%%%%%%%%%%%%%%%%%%%%%%%%%%%%%%%%%%%%%%%%%%%%%%%%%
%\section*{Acknowledgements}
%%%%%%%%%%%%%%%%%%%%%%%%%%%%%%%%%%%%%%%%%%%%%%%%%%%%%%%%%%%%%%%%%%%%%

%%%%%%%%%%%%%%%%%%%%
\end{document}